\documentclass[11pt,reqno]{article}
\usepackage[utf8]{inputenc}

\usepackage[T1]{fontenc}
\usepackage{microtype}
\usepackage{dsfont}

\usepackage[bookmarksnumbered,hypertexnames=false,colorlinks=true,linkcolor=blue,urlcolor=blue,citecolor=blue,anchorcolor=green,breaklinks=true,pdfusetitle]{hyperref}
\usepackage{aligned-overset,bbm,braket,mathrsfs,amsmath,amssymb,xcolor,amsthm,mathtools,fullpage,graphicx,enumitem,caption,booktabs,amsrefs,bm, commath, xparse, subcaption, diagbox, authblk}
\usepackage[longend,ruled,linesnumbered]{algorithm2e}

\usepackage{tikzscale}
\usepackage{pgfplots}
\pgfplotsset{compat=newest}
\usepgfplotslibrary{groupplots}
\usepgfplotslibrary{polar}
\usepgfplotslibrary{smithchart}
\usepgfplotslibrary{statistics}
\usepgfplotslibrary{dateplot}
\usepgfplotslibrary{ternary}
\usetikzlibrary{arrows.meta}
\usetikzlibrary{backgrounds}
\usepgfplotslibrary{patchplots}
\usepgfplotslibrary{fillbetween}
\pgfplotsset{%
layers/standard/.define layer set={%
    background,axis background,axis grid,axis ticks,axis lines,axis tick labels,pre main,main,axis descriptions,axis foreground%
}{grid style= {/pgfplots/on layer=axis grid},%
    tick style= {/pgfplots/on layer=axis ticks},%
    axis line style= {/pgfplots/on layer=axis lines},%
    label style= {/pgfplots/on layer=axis descriptions},%
    legend style= {/pgfplots/on layer=axis descriptions},%
    title style= {/pgfplots/on layer=axis descriptions},%
    colorbar style= {/pgfplots/on layer=axis descriptions},%
    ticklabel style= {/pgfplots/on layer=axis tick labels},%
    axis background@ style={/pgfplots/on layer=axis background},%
    3d box foreground style={/pgfplots/on layer=axis foreground},%
    },
}

\usepackage{multirow}
\usepackage{textcomp}
\usepackage[english]{babel}
\usepackage[capitalize]{cleveref}
\usepackage{thm-restate}
\newtheorem{thm}{Theorem}
\newtheorem{lem}[thm]{Lemma}\crefname{lem}{Lemma}{Lemmas}

\newtheorem{definition}[thm]{Definition}

\newtheorem{rem}[thm]{Remark}
\newtheorem{cor}[thm]{Corollary}

\DeclareMathOperator{\sign}{sign}

\DeclareMathOperator{\erf}{erf}

\DeclareMathOperator{\polylog}{polylog}

\renewcommand{\vec}[1]{\bm{#1}}
\newcommand{\mat}[1]{\bm{#1}}

\newcommand{\RR}{\mathbb R}
\newcommand{\N}{\mathbb N}
\newcommand{\R}{\mathbb R}
\newcommand{\C}{\mathbb C}

\newcommand{\eps}{\varepsilon}
\renewcommand{\epsilon}{\varepsilon}

\newcommand{\cheb}{\mathcal{T}}

\newcommand{\chebmat}{\bm{\mathcal{T}}}
\newcommand{\softO}{\widetilde{O}}
\renewcommand{\i}{\mathbf{i}}

\renewcommand{\dag}{*}
\renewcommand{\vec}[1]{\boldsymbol{\mathbf{#1}}}
\DeclarePairedDelimiterX{\ip}[2]{\langle}{\rangle}{#1,#2}
\newcommand{\ketbra}[2]{\ket{#1}\!\!\bra{#2}}

\allowdisplaybreaks[4]

\crefalias{AlgoLine}{line}
\crefname{line}{line}{lines}

\SetKwInput{Input}{Input}
\SetKwInput{Output}{Output}

\begin{document}
\date{}
\title{Improving quantum linear system solvers via a gradient descent perspective}

\author[1]{Sander Gribling}
\author[1, 2]{Iordanis Kerenidis}
\author[1]{D\'aniel Szil\'agyi}

\affil[1]{Université de Paris, CNRS, IRIF, F-75006, Paris, France}
\affil[2]{QC Ware, Palo Alto, USA and Paris, France}

\maketitle

\begin{abstract}
Solving systems of linear equations is one of the most important primitives in quantum computing that has the potential to provide a practical quantum advantage in many different areas, including in optimization, simulation, and machine learning. In this work, we revisit quantum linear system solvers from the perspective of convex optimization, and in particular gradient descent-type algorithms. This leads to a considerable constant-factor improvement in the runtime (or, conversely, a several orders of magnitude smaller error with the same runtime/circuit depth). 

More precisely, we first show how the asymptotically optimal quantum linear system solver of Childs, Kothari, and Somma is related to the gradient descent algorithm on the convex function $\|\mat A\vec x - \vec b\|_2^2$: their linear system solver is based on a truncation in the Chebyshev basis of the degree-$(t-1)$ polynomial (in $\mat A$) that maps the initial solution $\vec{x}_1 := \vec{b}$ to the $t$-th iterate $\vec{x}_t$ in the basic gradient descent algorithm. Then, instead of starting from the basic gradient descent algorithm, we use the optimal Chebyshev iteration method (which can be viewed as an accelerated gradient descent algorithm) and show that this leads to considerable improvements in the quantum solver. 
\end{abstract}

\section{Introduction}

The quantum linear systems (QLS) problem asks for a state that encodes the solution of a linear system $\mat A\vec{x}=\vec{b}$ where $\mat A \in \R^{n\times n}$ and $\vec{b} \in \R^n$.\footnote{Without loss of generality, one may assume that $\mat A$ is Hermitian.}
In a seminal work, Harrow, Hassidim, and Lloyd showed how to solve the QLS-problem using only $\polylog(n)$ queries to the input~\cite{HHL09}. Their algorithm has a polynomial dependence on the condition number $\kappa$ of $\mat A$ and the desired precision $\eps>0$. Subsequent work has improved the $\kappa$-dependence to near linear~\cite{Ambainis2012VTAA},\footnote{Here by a near linear runtime in terms of $\kappa$ we mean a runtime that scales as $\kappa \polylog(\kappa)$.} and the error-dependence to $\polylog(1/\eps)$~\cite{CKS17}. The algorithms in \cites{HHL09,Ambainis2012VTAA,CKS17} can all be viewed as implementing a polynomial transformation of $\mat A $ that approximates the inverse. They are based on various combinations of Hamiltonian simulation, quantum walks, linear combinations of unitaries, and most recently the quantum singular value transformation framework~\cite{Low2019hamiltonian,Gilyen2019SVT}.

More generally, an efficient QLS algorithm is a key building block for many downstream applications in optimization and in machine learning. Some examples include least-squares regression \cite{Chakraborty2019BE}, support-vector machines \cite{Rebentrost2014SVM,Kerenidis2021QSVM}, as well as differential equations \cite{linden2020quantum,tong2020fast}. Thus, optimizing the resources (depth in particular) required by the QLS algorithm would bring us closer to running these algorithms on near-term quantum hardware.

Currently, the best QLS algorithm is based on the polynomial by Childs, Kothari and Somma \cite{CKS17} (abbreviated as CKS from now on), which is evaluated by the quantum singular value transformation (QSVT) framework \cite{Gilyen2019SVT} by Gily\'{e}n, Su, Low and Wiebe, and sped-up using the variable-time amplitude amplification technique due to Ambainis \cite{Ambainis2012VTAA}. In a nutshell, the CKS polynomial is obtained by starting from the polynomial $p_t(x):= \frac{1-(1-x^2)^t}{x}$ for $t = \softO(\kappa^2)$, expressing it in the Chebyshev basis, and truncating the sum after $\softO(\kappa)$ terms. To obtain a quantum algorithm the resulting polynomial is then combined with either the linear combination of unitaries (LCU) lemma \cite{Berry2015Hamiltonian} or with the QSVT framework. The LCU approach is simpler but yields circuits with an extra multiplicative logarithmic factor in depth and an extra additive logarithmic number of ancillary qubits. On the other hand, the QSVT framework requires the computation of certain angles (see \cref{sec: QSVT} for details), and doing so efficiently in a numerically stable way is the subject of ongoing research \cites{Haah2019product, chao2020finding, dong2021efficient}. 

In this work, we revisit quantum linear system solvers by conceptually linking the previous and our techniques to classical gradient descent methods, and by providing the optimal quantum circuits within this framework. In particular, we point out a connection between classical iterative methods for solving $\mat A\vec{x}=\vec{b}$ and the polynomials used in the quantum algorithms, namely that they correspond to the ones used in the (basic) gradient descent on the convex function $\|\mat A\vec{x}-\vec{b}\|_2^2$. 

Our main contribution is to show that the optimal classical iterative method (known as \emph{Chebyshev iteration}) also leads to polynomials that can be implemented on a quantum computer, see \cref{thm:main}. This leads to a considerable constant-factor improvement in the runtime of QLS-solvers (or, conversely, an improved error with the same runtime/circuit depth).

In more detail, our approach is as follows. First recall that the Chebyshev iteration corresponds to the polynomial
\[
    q_t(x) := \frac{\left. \cheb_t\left(\frac{1+1/\kappa^2-2x^2}{1-1/\kappa^2}\right) \middle/ \cheb_t\left(\frac{1+1/\kappa^2}{1-1/\kappa^2}\right) \right.}{x},
\]
where $\cheb_t$ is the $t$-th Chebyshev polynomial of the first kind. These Chebyshev polynomials are defined as $\cheb_0(x) = 1$, $\cheb_1(x) = x$, and $\cheb_{t+1}(x) = 2x \cheb_t(x) -\cheb_{t-1}(x)$ for $t \geq 1$. They have the property that $|\cheb_t(x)|\leq 1$ for all $x \in [-1,1]$ and $t \geq 0$. One can show that the polynomial $q_t$ is an $\eps$-approximation of the inverse on the domain $x \in [-1, -1/\kappa] \cup [1/\kappa, 1]$, whenever $t \geq \frac12 \kappa \log(2\kappa^2 / \eps)$. To bound the maximum absolute value of $q_t$ on $[-1,1]$, we express $q_t(x)$ as $\sum_{i=0}^{t-1} c_i \cheb_{2i+1}(x)$ and bound the $1$-norm of the vector $\vec c$. The vector of coefficients can be used to implement $q_t(\mat A)/\|\vec c\|_1$ either directly via the linear combinations of unitaries approach, or via the quantum singular value transformation approach. In \cref{app:Good example functions for LCU} we show that this approach of bounding the $1$-norm of the vector of coefficients in the Chebyshev basis more generally leads to near optimal quantum algorithms via the LCU framework for a variety of continuous functions (powers of monomials, exponentials, logarithms) and discontinuous functions (the error function and by extension the sign and rectangle functions). 

The state of the art quantum linear systems solvers have a complexity that grows linearly in the condition number $\kappa$. In the small-$\kappa$ regime ($\kappa = O(n)$), it has long been known that $\Omega(\kappa)$ queries to the entries of the matrix are also needed for general linear systems~\cite{HHL09} and recently this bound has (surprisingly) been extended to the case of positive definite systems \cite{orsucci2021solving}. For larger $\kappa$ less is known. For example, we do not know if quantum algorithms can improve classical algorithms if $\kappa$ is large (i.e., can we beat matrix multiplication time?). We do not even have a linear lower bound: are $\Omega(n^2)$ queries needed when $\kappa = \Omega(n^2)$? In \cite{Dorn2009QueryComplexity} this question was answered positively when one wants to obtain a \emph{classical description} of $\mat A^{-1}\vec b$ and here we present a simplified proof of this result. 

\paragraph{Organization.} In \cref{sec:Convex optimization} we recall the gradient descent algorithm and elaborate its connection to the algorithm of~\cite{CKS17}. In \cref{sec:Chebyshev iteration} we discuss Chebyshev iteration, the optimal\footnote{We will make clear in what sense it is optimal in \cref{sec:Chebyshev iteration}} iterative method for solving linear systems. We show in \cref{sec:Quantum algorithm} that Chebyshev iteration lead to polynomials that can be efficiently implemented using for example the QSVT framework.
Finally, in \cref{sec:Query lower bound}, we give an overview of known lower bounds on the complexity of quantum linear system solvers both in the small $\kappa$ regime and in the large $\kappa$ regime. 

\section{Preliminaries}
\subsection{Polynomials and approximations}
\paragraph{Problem definition.} We consider linear systems that are defined by a Hermitian $n$-by-$n$ matrix $\mat A  \in \C^{n \times n}$ and a unit vector $\vec{b} \in \C^n$. We use $\kappa$ to denote the condition number of $\mat A $, that is, we assume that all non-zero eigenvalues of $\mat A $ lie in the set $D_\kappa := [-1,-1/\kappa] \cup [1/\kappa,1]$. Our goal is to \emph{approximately} solve the linear system 
\[
\mat A\vec{x} =\vec{b}.
\]
One can consider different notions of approximate solutions. Two natural ones are the following: 
\begin{enumerate}
    \item[1)] return $\vec{\tilde x}$ such that $\|\vec{\tilde x} - \mat A^{-1}\vec{b}\| \leq \eps$.
    \item[2)] return $\vec{\tilde x}$ such that $\|\mat A \vec{\tilde x} - \vec{b}\| \leq \eps$.
\end{enumerate}
Up to a change in $\eps$, the two notions are equivalent. Indeed, we have the chain of inequalities
\begin{equation}\label{eq:equivalence of QLSP definitions}
    \|\mat A\vec{x}-\vec{b}\| \leq \|\vec{x}-\mat A^{-1}\vec{b}\| \leq \kappa  \|\mat A\vec{x}-\vec{b}\|.
\end{equation}
We will focus on algorithms that achieve a polylogarithmic dependence in $\eps$. In \cref{sec:Chebyshev iteration} we construct the optimal degree-$t$ polynomial for approximation in the second notion, see \cref{def:QLSP}. Prior work \cite{CKS17, Chakraborty2019BE, Gilyen2019SVT} focused on the first notion of approximation, which is equivalent up to polylogarithm factors in the complexity. In \cref{sec: comparison} we show (numerically) that our polynomials also improve over prior work with respect to approximation in the first notion.

\paragraph{Polynomials.}
Given a polynomial $p(x) = \sum_{t=1}^T c_t x^t$ with coefficients $c_t \in \C$, and a Hermitian matrix $\mat A $, we define $p(\mat A) = \sum_{t=1}^T c_t \mat A^t$. If we let $\mat A  = \sum_{i=1}^n \lambda_i \vec{u}_i \vec{u}_i^*$ be the eigendecomposition of $\mat A $, then $p(\mat A) = \sum_{i=1}^n p(\lambda_i) \vec{u}_i \vec{u}_i^*$. 
 
\paragraph{Chebyshev decomposition.}
It is also useful to consider the \emph{Chebyshev decomposition} of $p(x)$, i.e., the decomposition
\[
    p(x) = \sum_{i=0}^t c_i \cheb_i(x)
\]
in the basis $\{ \cheb_0(x), \cheb_1(x), \dots, \cheb_{t}(x)\}$, for some vector $\vec c = (c_i)_{i \in \{0,\ldots,t\}}$ of coefficients. One can give an analytic expression for the coefficients $c_i$ using the fact that the Chebyshev polynomials are orthogonal with respect to the \emph{Chebyshev measure} which is defined in terms of the Lebesgue measure as $\dif \mu(x) = (1-x^2)^{-1/2}\, \dif x$. In other words, $c_i = \int_{-1}^1 \frac{p(x) \cheb_i(x)}{\sqrt{1-x^2}}\dif x$. Note that in practice this integral is rarely computed explicitly, as there exist efficient interpolation-based methods for computing the coefficient-vector $\vec{c}$ \cite{Gentleman1972DCT}.

\paragraph{Approximating the inverse.} We focus on methods to obtain a vector $\vec{\tilde x}$ that approximates $\mat A ^{-1}\vec{b}$ that are based on polynomials that approximate the inverse function $\lambda \mapsto \lambda^{-1}$ on the domain $[1/\kappa,1]$ (in the case of positive definite matrices) or $D_\kappa$ (in the general case). For example, let $\mat A$ be a Hermitian matrix with eigenvalues in $[1/\kappa,1]$ and let $p:\R \to \R$ be a polynomial such that $|p(\lambda)-\lambda^{-1}|\leq \eps$ for $\lambda \in [1/\kappa,1]$. Then $\vec{\tilde x} := p(\mat A) \vec{b}$ satisfies 
\[
\|\vec{\tilde x} - \mat A^{-1}\vec{b}\| = \|\sum_i (p(\lambda_i) -\lambda_i^{-1}) \vec{u}_i \vec{u}_i^* \vec{b}\| \leq \|\sum_i (p(\lambda_i) -\lambda_i^{-1}) \vec{u}_i \vec{u}_i^*\| \|\vec{b}\| \leq \eps \|\vec{b}\|.
\]

\subsection{Quantum preliminaries}
There exist different input models that one might consider when solving the linear system problem. 
In the standard case of a dense matrix $\mat A $, one might assume that all entries of $\mat A $ are already stored in memory. Alternatively, if $\mat A $ is sparse, sometimes it is more efficient to consider oracle access to its nonzero entries.  In the quantum setting, this \emph{sparse-access} model is particularly amenable to speedups. In the sparse-access model we assume that access to $\mat A $ is provided through two oracles
\[
    \mathcal O_\text{nz}: \ket{j, \ell} \mapsto \ket{j, \nu(j, \ell)} \text{ and } \mathcal O_A: \ket{j, k, z} \mapsto \ket{j, k, z \oplus A_{jk}},
\]
where $\nu(j, \ell)$ is the row index of the $\ell$th nonzero entry of the $j$th column. Many quantum algorithms can be phrased naturally in terms of a different input model called the \emph{block-encoding} model \cite{Low2019hamiltonian, Chakraborty2019BE}. (One can efficiently construct a block-encoding, given sparse access.)
\begin{definition}[Block encoding]
    Let $\mat A  \in \RR^{n \times n}$ be a Hermitian matrix, and let $N\in \N$ be such that $n = 2^N$, and let $\mu \geq 1$. The $(N + a)$-qubit operator $U_{\mat A}$ is a $(\mu, a)$-block-encoding of $\mat A $ if it satisfies $\mat A  = \mu(\bra{0}^{\otimes a}\otimes I) U_{\mat A} (\ket{0}^{\otimes a}\otimes I)$.
\end{definition}
For convenience, if we are not interested in the number of ancillary qubits $a$, we simply call $U_{\mat A}$ a $\mu$-block-encoding.
In what follows, we assume that we have access to $U_{\mat A}$, an (exact\footnote{Constructing exact block-encodings of arbitrary matrices $\mat A $ that are given in the sparse-access input model is a priori not possible with a finite gate set. Instead, one can construct a block-encoding of an approximation $\tilde{\mat A}$, by allowing an overhead in the circuit depth that is proportional to $\log(\|\mat A - \tilde{\mat A}\|)$.}) $(1, a)$-block-encoding of $\mat A $. The case of $\mu$-block-encodings with $\mu > 1$ can be reduced to the former by replacing our starting matrix with $\mat A /\mu$, that has eigenvalues in $D_{\mu\kappa}$. Furthermore, we assume that $\mat A $ is invertible, with eigenvalues in $D_\kappa$. 
Finally, we assume that we have access to $U_{\vec b}$, a unitary that (exactly) prepares the state $\ket{\vec{b}} = \vec{b}/\norm{\vec{b}}$ on input $\ket{\vec 0}$: $U_{\vec b} \ket{\vec 0} = \ket{\vec b}$.

We define the quantum linear system problem (QLSP) as follows:
\begin{definition}[Quantum linear systems]\label{def:QLSP}
    Let $\mat A  \in \R^{n\times n}$ be a Hermitian matrix with eigenvalues in $D_\kappa$, let $\vec{b} \in \R^n$, and let $\eps > 0$. Given a block-encoding $U_{\mat A}$ of $\mat A $ and a state preparation oracle $U_{\vec b}$, output a state 
    \[
    \ket{\phi} = \alpha \ket{0}\ket{\vec x} + \beta \ket{1}\ket{\psi}
    \]
    where $\|\ket{\mat A\vec x} - \ket{b}\| \leq \eps$, $\ket{\psi}$ is an arbitrary state, and $\alpha, \beta \in \C$ are such that  $|\alpha|^2 + |\beta|^2 = 1$ and $|\alpha|^2 \geq 2/3$.
\end{definition}
As mentioned before, the widely-used definition from the the literature \cites{CKS17, Chakraborty2019BE, Gilyen2019SVT} is equivalent to \cref{def:QLSP} up to a change in $\epsilon$. In this paper we use \cref{def:QLSP}, as our algorithm is optimal in this sense. In \cref{sec: comparison} we (numerically) show that our algorithm also improves over prior work with respect to the more widely used definition. 

Recent approaches for solving the QLS problem are based on applying a block-encoding of $p(\mat A)$ to $\ket{\vec{b}}$. In the next two sections we describe two ways of computing a block-encoding of $p(\mat A)$: through the QSVT framework, or by decomposing $p$ in the Chebyshev basis, computing each term individually, and combining the results using the linear combination of unitaries lemma (the LCU approach). 

\subsubsection{QSVT approach} \label{sec: QSVT}
The most straightforward way for evaluating a polynomial quantumly is through the quantum singular value transformation framework \cite{Gilyen2019SVT}. Using QSVT, one can directly evaluate any polynomial $p$ as long as its sup-norm is suitably bounded. Here the sup-norm of $p$ is defined as 
\[
    \norm{p}_\infty := \max_{x \in [-1,1]} |p(x)|.
\]
This is achieved by performing a series of rotations by angles $\vec{\Phi}=(\phi_1, \dots, \phi_t)$ on a single qubit, that induces a degree-$t$ polynomial transformation of the singular values of $\mat A $. Determining these angles efficiently in a numerically stable way is the subject of ongoing research \cites{Haah2019product, chao2020finding, dong2021efficient}. 
Below, we state a version of QSVT suitable for evaluating even and odd polynomials, since this is the case we are most interested in.
\begin{thm}[\cite{Gilyen2019SVT}*{Corollary 18}, for block-encodings]\label{thm:QSVT}
    Let $\mat A  \in \RR^{n \times n}$ be Hermitian, and let $U_{\mat A}$ be a 1-block-encoding of $\mat A $. 
    Let $\Pi = (\ketbra{0}{0})^{\otimes a} \otimes I$, and suppose that $p \in \R[x]$ is a degree-$t$ polynomial of parity-$(t \bmod 2)$ satisfying $\norm{p}_\infty \leq 1$. Then there exists a $\vec \Phi \in \R^t$ such that
    \begin{equation*}
        p(\mat A) = \begin{cases}
            \left( \bra{+}\otimes\Pi \right) \left( \ketbra{0}{0} \otimes U_{\vec \Phi} + \ketbra{1}{1} \otimes U_{-\vec \Phi} \right) \left( \ket{+}\otimes \Pi \right) & \text{ if } t \text{ is odd, and}\\
            \left( \bra{+}\otimes\Pi \right) \left( \ketbra{0}{0} \otimes U_{\vec \Phi} + \ketbra{1}{1} \otimes U_{-\vec \Phi} \right) \left( \ket{+}\otimes \Pi \right) & \text{ if } t \text{ is even,}
        \end{cases}
    \end{equation*}
    where $U_\Phi$ is defined as the phased alternating sequence
    \begin{equation*}
        U_{\vec \Phi} := \begin{cases}
            e^{\i\phi_1 (2 \Pi - I)}U_{\mat A} \prod_{j=1}^{(t-1)/2} \left( e^{\i\phi_{2j} (2\Pi - I)} U_{\mat A}^\dag e^{\i \phi_{2j+1}(2\Pi -I)}U_{\mat A} \right) & \text{ if } n \text{ is odd, and}\\
            \prod_{j=1}^{t/2} \left( e^{\i\phi_{2j-1} (2\Pi - I)} U_{\mat A}^\dag e^{\i \phi_{2j}(2\Pi -I)}U_{\mat A} \right) & \text{ if } n \text{ is even.}
        \end{cases}
    \end{equation*}
\end{thm}
Note that QSVT is fundamentally limited to evaluating polynomials that are bounded by $1$ in absolute value on $[-1, 1]$ (since the output is a unitary matrix). Approximations $p$ of $x^{-1}$ on $D_\kappa$ are inherently not bounded by $1$ on the interval $[-1,1]$: they are around $\kappa$ for $x=1/\kappa$. The QSVT framework allows us to evaluate $p(x)/M$ on $\mat A $ where $M$ is an upper bound on $\norm{p}_\infty$. This subnormalization reduces the success probability of for example a QSVT-based QLS-solver. It is thus important to obtain polynomial approximations $p$ that moreover permit a good bound $M$. 
\subsubsection{LCU approach}
 An alternative approach is based on the Linear Combinations of Unitaries (LCU) lemma \cite{Berry2015Hamiltonian}. It uses the fact that Chebyshev polynomials have a particularly nice vector of angles, which permits an efficient implementation of the LCU circuit. 

\begin{lem}[{\cite[Lem.~9]{Gilyen2019SVT}}] \label{lem: cheb angles}
    Let $\vec \Phi \in \R^t$ be such that $\phi_1 = (1-t)\frac\pi2$ and $\phi_i = \frac\pi2$ for $2 \leq i \leq t$. For this choice of $\vec \Phi$, the polynomial $p$ from \cref{thm:QSVT} is $\cheb_t$, the $t$-th Chebyshev polynomial of the first kind.
\end{lem}

\paragraph{Computing a single Chebyshev polynomial.}
We consider in more detail the above circuit for computing $\cheb_{2t+1}(\mat A)$ for a matrix $\mat A $ with a 1-block-encoding $U_{\mat A}$. Let $\Pi = \ketbra{0}{0} \otimes I$ be the same projector as in \cref{thm:QSVT} (we drop the exponent ${\otimes a}$ for convenience, or equivalently, we assume that the block-encoding $U_{\mat A}$ has a single auxiliary qubit). By \cref{lem: cheb angles} the unitary
\[
    U_{2t+1} = e^{-\pi \i t (2\Pi - I)}U_{\mat A} \prod_{j=1}^{t} \left( e^{\i \frac\pi2 (2\Pi - I)} U_{\mat A}^\dag e^{\i \frac\pi2(2\Pi -I)}U_{\mat A} \right)
\]
satisfies $(\bra{0}\otimes I) U_{2t+1} (\ket{0}\otimes I) = \cheb_{2t+1}(\mat A)$. 
We first simplify the above. Note that $2\Pi - I$ has eigenvalues $\pm 1$ and therefore $e^{-\pi \i t (2\Pi-I)} = (-1)^t I$ and $e^{\i  \frac{\pi}{2}(2\Pi-I)} = \i(2\Pi-I)$. This means that 
\begin{align*}
    U_{2t+1} &= (-1)^t U_{\mat A} \prod_{j=1}^{t} \left( \i(2\Pi-I) U_{\mat A}^\dag \i(2\Pi-I) U_{\mat A} \right) \\
    &= (-1)^t (\i)^{2t} U_{\mat A} \prod_{j=1}^{t} \left( (2\Pi-I)
    U_{\mat A}^\dag (2\Pi-I) U_{\mat A} \right) \\
    &= U_{\mat A} \prod_{j=1}^{t} \Big(\underbrace{ (2\Pi-I) U_{\mat A}^\dag (2\Pi-I) U_{\mat A} }_{=:W}\Big) = U_{\mat A} W^t. 
\end{align*}
In other words, $U_{2t+1}$ can be viewed as $t$ applications of the unitary $W$, followed by a single application of $U_{\mat A}$. The circuit for even Chebyshev polynomials $U_{2t}$ is very similar, and can be obtained from $U_{2t+1}$ by removing the final application of (left multiplication by) $U_{\mat A}$ -- however, since we are ultimately interested in implementing the inverse, an odd function, we do not describe the circuit in more detail.

\paragraph{Computing a linear combination of Chebyshev polynomials.}
Given the above circuit that computes block-encodings of $\cheb_{2k+1}(\mat A)$ for $k \geq 0$, the next step is to compute a block-encoding of linear combinations of the form
\begin{equation} \label{eq: linear comb}
    p(\mat A) = \sum_{i=0}^{t-1} c_{i} \cheb_{2i+1}(\mat A).
\end{equation}
This can be achieved using a version of the LCU algorithm due to~\cite{CKS17}. In particular, the key to an efficient implementation of the linear combination $\sum_{i=0}^{t-1} c_i U_{2i+1}$ is the efficient implementation of the operator $\sum_{i=0}^{t-1} \ketbra ii \otimes U_{2i+1}$, which we achieve by introducing an $l = (\lceil \log_2t \rceil + 1)$-qubit \emph{counter} register, and successively applying $W, W^2, W^4, \dots, W^{2^{l-1}}$ controlled on qubits $0, 1, \dots, l-1$ of the counter, followed by a single application of $U_{\mat A}$ at the end. In~\cite{CKS17}*{Theorem~4}, this circuit is analyzed for a specific polynomial-approximation of the inverse. The analysis naturally extends to arbitrary polynomials of the form~\eqref{eq: linear comb}.
\begin{thm}[based on \cite{CKS17}] \label{thm: LCU + Cheb}
    Let $\mat A $ be a Hermitian matrix with eigenvalues in $D_\kappa$, let $U_{\mat A}$ be its block-encoding, and let $U_{\sqrt{\vec{c}}}$ be a unitary that prepares the state $\frac{1}{\sqrt{\norm{\vec{c}}_1}} \sum_{i=1}^n \sqrt{c_i} \ket{i}$. Then, there exists an algorithm that computes a $\norm{\vec{c}}_1$-block-encoding of $p(\mat A)$ using $t+1$ calls to controlled versions of $U_{\mat A}$ and $U_{\mat A}^\dag$, and a single call to each of $U_{\sqrt{\vec{c}}}$ and $U_{\sqrt{\vec{c}}}^\dag$. This circuit uses a logarithmic number of additional qubits, and has a gate complexity of $O(t\polylog(nt\kappa/\eps))$.
\end{thm}
Compared to the QSVT approach, for this circuit we only need to compute the Chebyshev coefficients $\vec{c}$, as opposed to the vector of angles $\vec{\Phi}$ -- this comes, however, at the cost of using $O(\log t)$ additional qubits. Moreover, the coefficient 1-norm $\norm{\vec{c}}_1$ represents an upper bound for $\norm{p}_\infty$, since
\begin{equation}\label{eq:coeff norm is an upper bound for p on the interval}
    |p(x)| = \left| \sum_{i=0}^t c_i \cheb_i(x) \right| \leq \sum_{i=0}^t |c_i| \cdot |\cheb_i(x)| \leq \norm{\vec{c}}_1,\; \text{ for } |x| \leq 1.
\end{equation}
A natural question is how tight this bound is for general degree-$t$ polynomials $p$ with $\norm{p}_\infty \leq 1$. By norm conversion (\cref{eq:cheb ineq 1-norm} in particular), the ratio $\norm{\vec{c}}_1 / \norm{p}_\infty$ is provably upper bounded by $O(\sqrt{t})$ but in \cref{app:Good example functions for LCU} we observe that for many ``interesting'' functions the ratio $\norm{\vec{c}}_1 / \norm{p}_\infty$ is in fact only $O(\log(t))$.
A notable exception is the complex exponential $e^{\i \kappa x}$ (and thus $\sin(\kappa x)$ and $\cos(\kappa x)$) for which numerical experiments suggest that it attains the $O(\sqrt{t})$ upper bound.

\section{Convex optimization perspective}\label{sec:Convex optimization}
In this section we introduce the convex optimization approach to linear system solving, and reinterpret the CKS polynomial in this framework. Let us first assume that $\mat A$ is positive definite (PD). We start by defining the convex function $f:\R^n \to \R$ 
as
\[
f(x) := \frac{\vec{x}^\top \mat A \vec{x}}{2} - \vec{b}^\top \vec{x}.
\]
Note that $\nabla f(\vec{x}) = \mat A\vec{x} - \vec{b}$, so the minimizer of $f$ satisfies $\mat A \vec{x} =\vec{b}$. This observation forms the basis of the convex optimization approach to linear system solving. We refer the reader to, for example,~\cite{Polyak1987Optimization,Vishnoi2013Laplacian} for an overview of gradient descent type algorithms for solving linear systems. 

\subsection{Gradient descent}\label{sec:Gradient descent}
One of the most well-known algorithms for minimizing a convex function $f$ is the family of gradient descent algorithms. Starting from an initial point $\vec{x}_1$ (we use 1-based indexing on purpose), such an algorithm performs the iterations
\begin{equation} \label{eq:GD}
\vec{x}_t = \vec{x}_{t-1} - \eta_t \nabla f(\vec{x}_{t-1}), \qquad \text{where }t = 2,3,\ldots
\end{equation}
where $\eta_t \in [0,\infty)$ is the `step size' in the $t$-th iteration. For the most basic version of gradient descent we take a constant step size, i.e., $\eta_t$ is independent of $t$. 

For our quadratic function $f$ we can unpack this recurrence. As we have seen before $\nabla f(\vec{x}_{t-1}) = \mat A\vec{x}_{t-1} - \vec{b}$ and hence 
\[
\vec{x}_{t} = (I-\eta_t \mat A) \vec{x}_{t-1} + \eta_t \vec{b}.
\]
If we set $\eta_t := 1$ for all $t \in \N$ and use the initial point $\vec{x}_1 = \vec{b}$ (or even $\vec{x}_0=\vec{0}$, if we allow empty sums), then we obtain 
\[
\vec{x}_t = \sum_{k=0}^{t-1} (I-\mat A)^k \vec{b}.
\]
Let us define the polynomial $p_t^+(\lambda) = \sum_{k=0}^{t-1} (1-\lambda)^k$ so that $\vec{x}_t = p_t^+(\mat A)\vec{b}$. Observe that this is the degree-$(t-1)$ Taylor expansion of the function $1/\lambda$ around $1$. 
\begin{lem} \label{lem:pt}
We have $|p_t^+(\lambda) - 1/\lambda| \leq \eps$ for all $\lambda \in [1/\kappa,1]$ whenever $t \geq \kappa \log(\kappa/\eps)$. 
\end{lem}
\begin{proof}
Indeed, substituting $\delta = 1-\lambda$ we have $(1-\delta) \cdot p_t^+(1-\delta) = (1-\delta) \sum_{k=0}^{t-1} \delta^k = 1-\delta^t$ which shows that 
\[
p_t^+(\lambda) = \frac{1-(1-\lambda)^t}\lambda.
\]
Therefore, for $\lambda \in [1/\kappa,1]$ and $t \geq \kappa \log(1/\eps)$ we have
\[
|p_t^+(\lambda) - 1/\lambda| = |1/\lambda|\cdot|1-\lambda|^t \leq \kappa (1-1/\kappa)^t \leq \kappa e^{-\log(\kappa/\epsilon)} = \epsilon. \qedhere
\]
\end{proof}

\subsection{Chebyshev iteration} \label{sec:Chebyshev iteration}
In the previous section we saw that $t$-step iterative methods are roughly equivalent to degree-$(t-1)$ polynomials that approximate the function $1/x$ on $[1/\kappa, 1]$. Thus, the natural question to ask is what is the best such polynomial $q_t^+$? Here we use the notion of optimality that comes from \cref{def:QLSP}. In other words, what is the  degree-$(t-1)$ polynomial $q_t^+$ that minimizes 
\begin{equation}\label{eq:sense of optimality}
    \max_{x \in [1/\kappa, 1]} |xq_t^+(x) - 1|.
\end{equation}
First observe that all such polynomials can be expressed in the form $q_t^+(x) = \frac{1 - r_t^+(x)}{x}$ where $r_t^+$ is a degree-$t$ polynomial that satisfies $r_t^+(0) = 1$.\footnote{For example, in the case of gradient descent we have $r_t^+(x) = (1-x)^t$.} Thus, our goal is to find a degree-$t$ polynomial $r_t^+(x)$ that has the smallest absolute value on the interval $[1/\kappa, 1]$ and satisfies the normalization constraint $r_t^+(0)=1$. It turns out that we can use extremal properties of the Chebyshev polynomials $\cheb_t(x)$ to determine an optimal $r_t^+(x)$. We use the following well-known result (cf.~\cite[Prop.~2.4]{Vishnoi2013Approximation}). 
\begin{lem} \label{thm: Chebyshev extremality}
    For any degree-$t$ polynomial $p(x)$ such that $|p(x)| \leq 1$ for all $x \in [-1, 1]$, and any $y$ such that $|y| > 1$, we have $|p(y)| \leq |\cheb_t(y)|$.
\end{lem}
Using the affine transformation $x \mapsto \frac{1+1/\kappa-2x}{1-1/\kappa}$ this gives the following corollary:
\begin{cor} \label{cor: cheb functions}
    Let $\kappa > 1$ be real, and let $t > 0$ be an integer. Then, the polynomial 
    \[
        r_t^+(x) = \left. \cheb_t\left(\frac{1+1/\kappa-2x}{1-1/\kappa}\right) \middle/ \cheb_t\left(\frac{1+1/\kappa}{1-1/\kappa}\right) \right.
    \]
    is a degree-$t$ polynomial that satisfies $r_t^+(0) = 1$, and minimizes the quantity $\max_{x \in [1, 1/\kappa]} |r_t^+(x)|$.
\end{cor} 
Note that the polynomials $r_t^+$ satisfy a Chebyshev-like $3$-term recurrence. As a consequence, the polynomials $q_t^+(x) = (1-r_t^+(x))/x$ also satisfy such a recurrence. The corresponding iterative method is known as the Chebyshev iteration.
\begin{rem}[Chebyshev iteration] The polynomial $q_t^+(x)$ satisfies the recurrence
\begin{equation} \label{eq:qt recurrence}
    q_{t+1}^+(x) = 2 \frac{\cheb_t(\gamma)}{\cheb_{t+1}(\gamma)} \frac{\kappa + 1 - 2\kappa x}{\kappa - 1} q_t^+(x) - \frac{\cheb_{t-1}(\gamma)}{\cheb_{t+1}(\gamma)}q_{t-1}^+(x) - \frac{4\kappa}{\kappa - 1}\frac{\cheb_t(\gamma)}{\cheb_{t+1}(\gamma)},
\end{equation}
where $\gamma = \frac{1+1/\kappa}{1-1/\kappa}$.
This recurrence corresponds to the iterative method $\vec{x}_1 = \vec{b}$ and 
\[
    \vec{x}_{t+1} = 2 \frac{\cheb_t(\gamma)}{\cheb_{t+1}(\gamma)} \frac{(\kappa + 1)I - 2\kappa \mat A}{\kappa - 1} \vec{x}_t - \frac{\cheb_{t-1}(\gamma)}{\cheb_{t+1}(\gamma)}\vec{x}_{t-1} - \frac{4\kappa}{\kappa - 1}\frac{\cheb_t(\gamma)}{\cheb_{t+1}(\gamma)}\vec{b}.
\]
\end{rem}
The convergence rate of this method is summarized by the following theorem:
\begin{thm} \label{thm: cheb approx pos}
Let $\kappa > 1$ and $\epsilon > 0$. Then, for $t \geq \frac12 \sqrt{\kappa} \log(2\kappa / \epsilon)$ we have
\[
    \left\lvert q_t^+(x) - 1/x \right \rvert \leq \epsilon \text{ for all } x \in [1/\kappa, 1].
\]
\end{thm}
\begin{proof}
    First, we define $s(x) = \frac{1+1/\kappa-2x}{1-1/\kappa}$, so that we have $r_t^+(x) = \cheb_t(s(x))/\cheb_t(s(0))$. Thus, for all $x \in [1/\kappa, 1]$, we have 
    \[
        |q_t^+(x) - 1/x| = |r_t^+(x) / x| \leq \kappa |r_t^+(x)| = \kappa \left\lvert \cheb_t(s(x))/\cheb_t(s(0)) \right \rvert.
    \]
    Additionally, since $|s(x)| \leq 1$ on this interval, we also have $|\cheb_t(s(x))|\leq 1$. Thus, it suffices to find $t$ for which $\cheb_t(s(0)) = \cheb_t(1+ \frac{2}{\kappa - 1}) \geq \frac{\kappa}{\epsilon}$. Since the Chebyshev polynomial $\cheb_t(\cdot)$ can be computed as
    \begin{equation}\label{eq:explicit expression for Tn for big x}
        \cheb_t(x) = \frac12 \left( \left( x - \sqrt{x^2 - 1} \right)^t + \left( x + \sqrt{x^2 - 1} \right)^t \right) \text{ for } |x| \geq 1,
    \end{equation}
    we can conclude that $\cheb_t(s(0)) = \frac12 \left( \left(\frac{\sqrt\kappa - 1}{\sqrt\kappa+1}\right)^t + \left(\frac{\sqrt\kappa+1}{\sqrt\kappa-1}\right)^t \right) \geq \frac12 \left(\frac{\sqrt\kappa+1}{\sqrt\kappa-1}\right)^t$. Using the inequality $(1+\frac{x}{n})^{n+x/2} \geq e^x$ for $x,n \geq 0$, after substituting $t = \frac12 \sqrt{\kappa} \log(2\kappa / \epsilon)$ we have
    \begin{align*}
        \cheb_t(s(0)) &\geq \frac12\left(\frac{\sqrt\kappa+1}{\sqrt\kappa-1}\right)^t = \frac12\left( 1 + \frac{2}{\sqrt{\kappa} - 1} \right)^{(\sqrt{\kappa}-1 + 2/2)\frac{\log(2\kappa/\epsilon)}{2}} \\
        &\geq \frac12 \exp(\log(2\kappa/\epsilon)) = \frac{\kappa}{\epsilon}. \qedhere
    \end{align*}
\end{proof}

\subsection{The general case}

We now return to the setting where $\mat A$ is a Hermitian matrix and has eigenvalues in the domain $D_\kappa = [-1,-1/\kappa] \cup [1/\kappa,1]$. One can still solve such systems using gradient descent methods by reducing to the convex case. That is, by considering the equivalent linear system $\mat A ^2 \vec{x} = \mat A \vec{b}$ and the corresponding convex function $f(\vec{x}) = \frac12 \vec{x}^\top \mat A^2 \vec{x} - \vec{b}^\top \mat A \vec{x}$. In particular, this allows us to solve $\mat A\vec x = \vec b$ by using a method for solving PD systems applied to the system $\mat A ^2 \vec{x} = \mat A \vec{b}$. 
\begin{cor}\label{cor:reducing general linear systems to psd}
    Let $\eps > 0$, $\kappa > 1$, and let $P_t$ be any degree-$(t-1)$ polynomial such that $|P_t(\lambda) - 1/\lambda| \leq \eps$ for all $\lambda \in [1/\kappa^2, 1]$. Then, $|P_t(\mu^2)\mu - 1/\mu| \leq \epsilon$ for all $\mu \in D_\kappa$. 
\end{cor}
\begin{proof}
    Since $|\mu| \leq 1$ on $D_\kappa$, we have
$|P_t(\mu^2)\mu - 1/\mu| = |\mu|\cdot|P_t(\mu^2) - 1/\mu^2| \leq \eps$.
\end{proof}
We define the following two polynomials as the respective analogs of $p_t^+$ and $q_t^+$ for $D_\kappa$:
\begin{align}
    p_t(x) &= x p_t^+(x^2) = \frac{1 - (1-x^2)^t}{x}, \text{ and}\\
    q_t(x) &= x q_t^+(x^2) = \frac{1 - \cheb_t(\frac{1+1/\kappa^2 - 2x^2}{1-1/\kappa^2}) / \cheb_t(\frac{1+\kappa^2}{1-\kappa^2})}{x}. \label{eq: def qt}
\end{align}
Both $p_t$ and $q_t$ are degree-$(2t-1)$ polynomials, but different values of $t$ are required in order to achieve an $\eps$-approximation of $1/x$ on $D_\kappa$. In particular, the following degrees are required:
\begin{cor}\label{cor:degrees of pt and qt}
    Let $\kappa > 1$ and $\eps > 0$. Then, 
    \begin{enumerate}
        \item $|p_t(x) - 1/x| \leq \eps$ for all $x \in D_\kappa$ whenever $t \geq \kappa^2 \log(\kappa^2 / \eps)$,
        \item $|q_t(x) - 1/x| \leq \eps$ for all $x \in D_\kappa$ whenever $t \geq \frac12 \kappa \log(2\kappa^2 / \eps)$.
    \end{enumerate}
\end{cor}

\begin{lem}\label{cor:chebyshev iteration is optimal for general matrices}
    Let $t \in \N$ and $\kappa > 1$. The polynomial $q_t$ is a degree-$(2t-1)$ polynomial that minimizes the quantity $\max_{x \in D_\kappa} |x P(x) - 1|$ among all degree-$(2t-1)$ polynomials $P \in \R[x]$. 
\end{lem}
\begin{proof}
For a given $t$, we define 
    \[
    \eps^+ := \min_{\substack{P^+ \in \R[y]\\\operatorname{deg} P^+ = t-1}}\max_{y \in [1/\kappa^2, 1]} |y P^+(y) - 1|, \qquad \eps := \min_{\substack{P \in \R[x]\\\operatorname{deg} P = 2t-1}}\max_{x \in D_\kappa} |x P(x) - 1|.
    \]
We first show that $q_t$ certifies that $\eps \leq \eps^+$, and then we show $\eps=\eps^+$. From \cref{sec:Chebyshev iteration}, we know that $\eps^+$ is achieved by the degree-$t-1$ polynomial $q_t^+(x) := \frac{1 - \cheb_t(s(x)) / \cheb_t(s(0))}{x}$, where $s(x) := \frac{1+1/\kappa^2 - 2x}{1-1/\kappa^2}$. Then, for $q_t(x) := \frac{1 - \cheb_t(s(x^2)) / \cheb_t(s(0))}{x}$ we have 
    \[
        \max_{x \in D_\kappa} |x q_t(x) - 1| = \max_{x \in D_{\kappa}} |x^2 q_t^+(x^2) - 1| = \max_{y \in [1/\kappa^2,1]} |y q_t^+(y) - 1| = \eps^+,
    \]
    where in the first equality we use \cref{eq: def qt}. We now show that $\eps = \eps^+$. Let $P(x)$ be a degree-$(2t-1)$ polynomial that satisfies $\max_{x \in D_\kappa} |xP(x)-1| = \eps$. We first show that $P$ is odd. To do this, decompose $P$ as $P(x) = P_{\mathrm{even}}(x) + P_{\mathrm{odd}}(x)$ where $P_{\mathrm{even}}$ is even and $P_{\mathrm{odd}}$ is odd. Then 
    \begin{align*}
        \max_{x \in D_\kappa} |x P(x) - 1| &= \max_{x \in [1/\kappa, 1]} \max\{ |xP(x) - 1|, |-xP(-x) - 1| \} \\
        &= \max_{x \in [1/\kappa, 1]} \max\{ |xP_\mathrm{odd}(x) + xP_\mathrm{even}(x) - 1|, |xP_\mathrm{odd}(x) - xP_\mathrm{even}(x) - 1| \} \\
        &\geq \max_{x \in [1/\kappa, 1]} |xP_\mathrm{odd}(x) - 1| = \max_{x \in D_\kappa} |xP_\mathrm{odd}(x) - 1|.
    \end{align*}
    Hence replacing $P$ by $P_\mathrm{odd}$ decreases $\eps$, so we may assume that $P(x)$ is odd. Then $P(x)/x$ is a degree-$(2t-2)$ even polynomial. Let $P^+(y)$ be the degree-$(t-1)$ polynomial for which $P(x)/x = P^+(x^2)$. Then we have 
    \[
    \max_{y \in [1/\kappa^2,1]} |yP^+(y) - 1| = \max_{x \in [1/\kappa,1]} |x^2 P^+(x^2) - 1| = \max_{x \in D_\kappa} |x P(x) - 1| = \eps
    \]
This shows that $\eps^+ \leq \eps$ which concludes the proof: $q_t$ is the degree-$(2t-1)$ polynomial that minimizes $\max_{x \in D_\kappa} |xP(x)-1|$ over polynomials of degree $2t-1$. 
\end{proof}

\subsection{Relation to the Chebyshev approach of Childs-Kothari-Somma} \label{sec: comparison}
In~\cite{CKS17}, Childs, Kothari, and Somma approached the quantum linear system solver-problem by approximating the function $1/x$ on the domain $D_\kappa$ by (low-degree) polynomials. To start, they approximate $1/x$ by a function that is bounded near the origin: they multiply $1/x$ by a function that is small at the origin and close to 1 on $D_\kappa$. A natural choice for such a function is $1 - (1-x^2)^t$, so the function they end up with turns out to be exactly $p_t(x)$, the polynomial corresponding to $t$ steps of gradient descent applied to the quadratic $\frac12 \vec{x}^\top \mat A^2 \vec{x} - \vec{b}^\top \mat A \vec{x}$! So indeed, this is a good approximation of $1/x$ whenever $t \geq \kappa^2 \log(\kappa^2/\eps)$. 

The polynomial $p_t$ can be written in the Chebyshev basis as follows:
\begin{equation}\label{eq:CKS polynomial expanded}
    p_t(x) = 4 \sum_{j=0}^{t-1} (-1)^j \left(\frac{\sum_{i=j+1}^t \binom{2t}{t+i}}{2^{2t}}\right) \cheb_{2j+1}(x).
\end{equation}
The key insight of~\cite{CKS17} is that this expansion can be truncated at $\tilde O(\kappa)$ terms, since the Chebyshev coefficients decay exponentially. This can be shown by relating the absolute value of the $j$-th coefficient (for $j=0,1,\dots$) to the probability of more than $t+j$ heads appearing in $2t$ tosses of a fair coin. This probability decreases as $e^{-j^2/t}$ which can be seen by applying the Chernoff bound. Thus, starting from $p_t$, an $\eps$-approximation of the inverse, we obtain an $\eps$-approximation of $p_t$ by truncating the summation at $j = \sqrt{t \log(4t / \eps)} = \tilde O(\kappa)$ -- so, for these parameters, the CKS polynomial is a $2\eps$-approximation of the inverse on $D_\kappa$.

Although this truncated polynomial is asymptotically optimal, it is not an optimum of \eqref{eq:sense of optimality}. Hence, the Chebyshev iteration polynomial provides a better approximation for a fixed degree, or conversely requires a lower degree to reach the same error on $D_\kappa$. In \cref{fig:degree comparison table}, we use \cref{cor:degrees of pt and qt} to compute the degree required to achieve error $\eps$ on $D_\kappa$, and observe that the degree of the CKS polynomial is roughly twice the degree of the corresponding Chebyshev iteration polynomial.
\begin{table}[h]
    \centering
    \subfloat[][CKS polynomial]{
        \begin{tabular}{c|c|c|c|c}
         \diagbox{$\kappa$}{$\eps$}	& 0.5 & $10^{-2}$ & $10^{-4}$ & $10^{-6}$ \\\hline
            2 & 15 & 33 & 53 & 71 \\ \hline
            10 & 115 & 203 & 301 & 399 \\ \hline
            100 & 1819 & 2687 & 3669 & 4633 \\ \hline
            1000 & 24913 & 33515 & 43337 & 52989
        \end{tabular}
    }
    \subfloat[][Chebyshev iteration]{
        \begin{tabular}{c|c|c|c|c}
        	\diagbox{$\kappa$}{$\eps$}	& 0.5 & $10^{-2}$ & $10^{-4}$ & $10^{-6}$ \\\hline
        	    2 & 7 & 15 & 25 & 33 \\ \hline
                10 & 61 & 101 & 147 & 193 \\ \hline
                100 & 1061 & 1453 & 1913 & 2373 \\ \hline
                1000 & 15203 & 19115 & 23721 & 28327
        \end{tabular}
    }
    \caption{Degrees of approximation polynomials for a given condition condition number $\kappa$ and error $\eps$, computed according to \cref{cor:degrees of pt and qt}.}
    \label{fig:degree comparison table}
\end{table}
\begin{figure}[h]
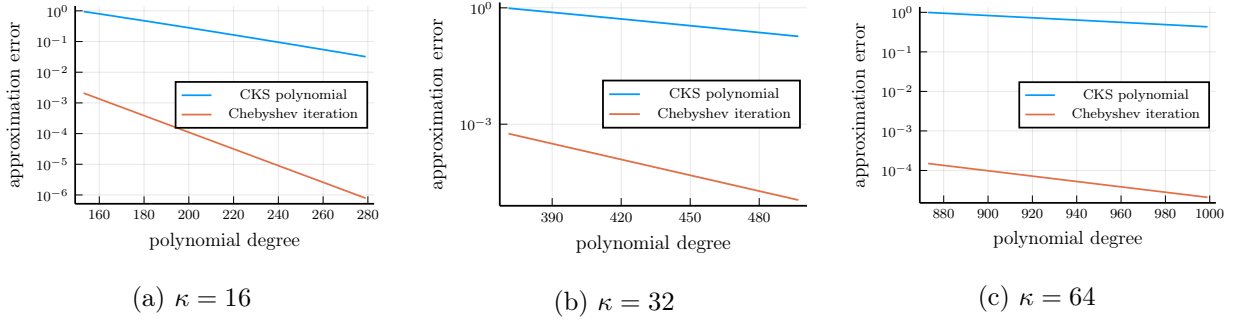

    \centering
    \begingroup
    \pgfplotsset{every axis/.style={scale=0.4}}
	\begin{subfigure}[h]{0.33\textwidth}
		\centering
		\includegraphics[width=\linewidth]{error_kappa=16_degree=var.tikz}
		\caption{$\kappa = 16$}
	\end{subfigure}~
	\begin{subfigure}[h]{0.33\textwidth}
		\centering
		\includegraphics[width=\linewidth]{error_kappa=32_degree=var.tikz}
		\caption{$\kappa = 32$}
	\end{subfigure}~
	\begin{subfigure}[h]{0.33\textwidth}
		\centering
		\includegraphics[width=\linewidth]{error_kappa=64_degree=var.tikz}
		\caption{$\kappa = 64$}
	\end{subfigure}
	\endgroup
    \caption{Error comparison between Chebyshev iteration and the truncated gradient descent polynomial, for a fixed condition number $\kappa$ and varying degrees.}
    \label{fig:error wrt degree}
\end{figure}
\begin{figure}[h]
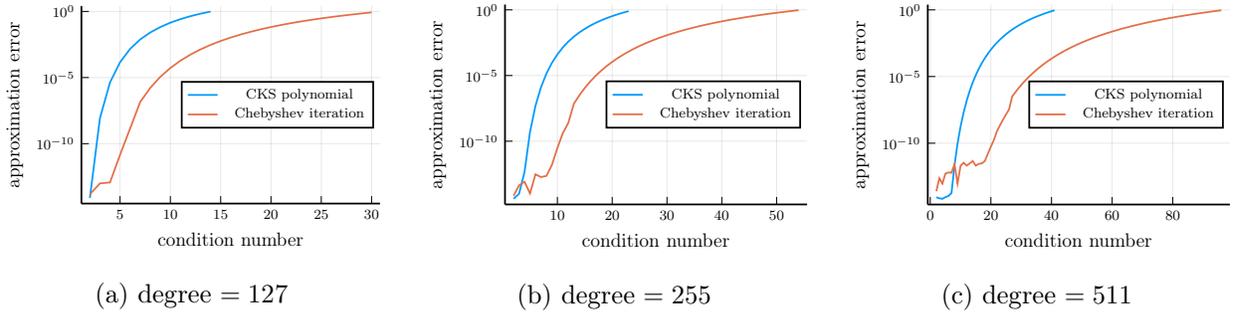

    \centering
    \begingroup
    \pgfplotsset{every axis/.style={scale=0.4}}
	\begin{subfigure}[h]{0.33\textwidth}
		\centering
		\includegraphics[width=\textwidth]{error_kappa=var_degree=127.tikz}
		\caption{degree $=127$}
	\end{subfigure}~
	\begin{subfigure}[h]{0.33\textwidth}
		\centering
		\includegraphics[width=\textwidth]{error_kappa=var_degree=255.tikz}
		\caption{degree $=255$}
	\end{subfigure}~
	\begin{subfigure}[h]{0.33\textwidth}
		\centering
		\includegraphics[width=\textwidth]{error_kappa=var_degree=511.tikz}
		\caption{degree $=511$}
	\end{subfigure}
	\endgroup
    \caption{Error comparison between Chebyshev iteration and the truncated gradient descent polynomial, for a fixed degree and varying condition numbers.}
    \label{fig:error wrt kappa}
\end{figure}
In Figures \ref{fig:error wrt degree} and \ref{fig:error wrt kappa} we compute the actual errors achieved by the two polynomials, for a given degree and condition number. In particular, in \cref{fig:error wrt degree} we see that that for a fixed condition number, the convergence is linear for both polynomials, with a faster rate of convergence in case of Chebyshev iteration (so, for the same degree, the difference in errors is a few orders of magnitude). Further numerical experiments indicate that the ratio of the convergence rates (i.e. the slopes of the lines on \cref{fig:error wrt degree}) is roughly 2, independently of the choice of $\kappa$. Conversely, in \cref{fig:error wrt kappa} we see that with circuits of fixed depth, the error of Chebyshev iteration is an order of magnitude lower, no matter the condition number (in the figure we only consider polynomials that achieve an error $\eps \leq 1$). 

\section{A quantum algorithm}\label{sec:Quantum algorithm}
As mentioned before, our algorithm can be described as applying the polynomial $q_t$ to a 1-block-encoding of the input matrix $\mat A $. This yields an $O(t)$-block-encoding of $q_t(\mat A)$, which can then be applied to the input state $\ket{\vec{b}}$. Formally, we show the following. 
\begin{thm}[Main result]\label{thm:main}
	Let $\mat A $ be a Hermitian matrix with eigenvalues in $D_\kappa$, let $U_{\mat A}$ be a 1-block-encoding of $\mat A $, and let $\eps > 0$. Then, for $t \geq \frac12 \kappa \log(2\kappa^2 / \eps)$, a $2(1+\eps/\kappa^2)t$-block-encoding of $q_t(\mat A)$ can be constructed using $2t-1$ calls to $U_{\mat A}$ and $U_{\mat A}^\dag$.
\end{thm}
\begin{proof}
	The algorithm consists of applying QSVT (\cref{thm:QSVT}) to the polynomial $q_t(x)/\norm{q_t}_\infty$. This allows us to construct a $\norm{q_t}_\infty$-block-encoding of $q_t(\mat A)$ with the desired complexity. It remains to upper bound $\norm{q_t}_\infty$ by $2(1+\eps/\kappa^2)t$. Motivated by \cref{eq:coeff norm is an upper bound for p on the interval}, it suffices to upper bound the $1$-norm of the vector $\vec c$ of coefficients of $q_t$ in the Chebyshev basis (again by $2(1+\eps/\kappa^2)t$). In \cref{cor:bound for ct norm} we show that $\|\vec c_t\|_1 \leq 2(1+\eps/\kappa^2)t$.
\end{proof}
The block-encoding of $q_t(\mat A)$ can now be used as a black-box replacement for the block-encoding of the corresponding CKS polynomial evaluated at $\mat A $. For example, using variable-time amplitude amplification, an $\softO(\kappa)$-query (to $U_{\mat A}$) complexity QLS algorithm can be derived. We refer the reader to \cite{CKS17,Gilyen2019SVT,martyn2021grand} for an overview of these techniques.

As an alternative approach, one could use the fact that $\|\vec c_t\|_1$ is bounded in order to evaluate $q_t$ via LCU (\cref{thm: LCU + Cheb}). At the cost of using $O(\log t)$ additional qubits, an LCU-based approach would yield a more ``natural'' quantum algorithm, that does away with the classical angle computation preprocessing step required by QSVT -- computing these angles efficiently in a numerically stable way is the subject of ongoing research \cite{chao2020finding, dong2021efficient, Haah2019product}. Moreover, in \cref{app:Good example functions for LCU} we consider some other commonly-used functions, and bound their coefficient norms using similar techniques. For these functions, the coefficient norm is only a logarithmic factor away from the maximum absolute value on the interval $[-1, 1]$, meaning that they can be approximately evaluated with LCU in addition to QSVT, with slightly deeper circuits (multiplicative logarithmic overhead) and slightly more qubits (additive logarithmic overhead).

\subsection{Bounding the Chebyshev coefficients}\label{sec:Bounding the coeffs}
As discussed above, in order to apply (a normalized version of) $q_t$ to a block-encoding of a Hermitian matrix with eigenvalues in $D_\kappa$, we need a bound on the sup-norm of $q_t$ on the interval~$[-1,1]$. In order to derive such a bound, we express $q_t$ in the basis of Chebyshev polynomials. Each of the Chebyshev polynomials has sup-norm equal to $1$ and therefore a bound on the $1$-norm of the coefficient vector provides a bound on the sup-norm of $q_t$. Recall that since $q_t$ is an odd polynomial, its expansion in the Chebyshev basis only involves the odd-degree Chebyshev polynomials. That is, we can write 
\begin{equation} \label{eq:def ct}
q_t(x) = \sum_{i=0}^{t-1} c_{t,i} \cheb_{2i+1}(x)
\end{equation}
for some vector $\vec c_t = (c_{t,i})_{i \in \{0,\ldots,t-1\}}$ of coefficients. One can give an analytic expression for $c_{t,i}$ using the fact that the Chebyshev polynomials are orthogonal with respect to the \emph{Chebyshev measure}.
Here we take a different approach and use the following discrete orthogonality relations. Fix a degree $m \in \N$ and let $\{x_1,\ldots,x_{m}\}$ be the roots of $\cheb_m(x)$. The $x_k$'s are called the \emph{Chebyshev nodes} and they admit an analytic formula:
\begin{equation}
    x_k = \cos\left(\frac{(k-\tfrac12)\pi}{m}\right) \qquad \text{for } k=1,\ldots,m
\end{equation}
The discrete orthogonality relation that we will use is the following. For $0 \leq i,j < m$, we have
\begin{align}\label{eq:chebyshev discrete orthogonality}
    \sum_{k=1}^m \cheb_i(x_k) \cheb_j(x_k) = \begin{cases}
        m & \text{ if } i=j=0, \\
        \frac{m}{2} & \text{ if } i=j<m, \\
        0 & \text{ if } i \neq j.
    \end{cases}
\end{align}
Since $q_t$ is a polynomial of degree $2t-1$, we will use the discrete orthogonality conditions corresponding to $m = 2t$ to recover the coefficient of $\cheb_{2i+1}$ in $q_t$. We have 
\begin{equation}
c_{t,i} = \frac{1}{t} \sum_{k=1}^{2t} q_t(x_k) \cheb_{2i+1}(x_k)
\end{equation}
for all $i \in \{0,1,\ldots,t-1\}$. We can equivalently write this in matrix form, $\vec{c}_t = \frac1t \chebmat_t \vec{q}_t$, where
\[
    \chebmat_t = \begin{bmatrix}
        \cheb_1(x_1) & \cheb_1(x_2) & \dots & \cheb_1(x_{2t}) \\
        \cheb_3(x_1) & \cheb_3(x_2) & \dots & \cheb_3(x_{2t}) \\
        \vdots & \vdots & \ddots & \vdots \\
        \cheb_{2t-1}(x_1) & \cheb_{2t-1}(x_2) & \dots & \cheb_{2t-1}(x_{2t})
    \end{bmatrix} \text{ and } \vec{q}_t = \begin{bmatrix}
        q_t(x_1) \\
        q_t(x_2) \\
        \vdots \\
        q_t(x_{2t})
    \end{bmatrix}.
\]
Our goal is to show that $\|\vec c_t\|_1 \leq C\cdot t$ for a small constant $C$. To do so, we first use the Cauchy-Schwarz inequality to obtain 
\begin{equation} \label{eq:cheb ineq 1-norm}
\|\vec c_t\|_1 \leq \sqrt{t} \|\vec c_t\|_2 = \frac{1}{\sqrt{t}} \|\vec \chebmat_t \vec q_t\|_2 =  \frac{\|\chebmat_t\|}{\sqrt{t}} \|\vec q_t\|_2  = \|\vec q_t\|_2
\end{equation}
where the second to last equality follows from the discrete orthogonality relations \cref{eq:chebyshev discrete orthogonality}: we see that $\chebmat_t \chebmat_t^* = t I_t$ and therefore $\|\chebmat_t\| = \sqrt{t}$. We are thus left to bound $\|\vec q_t\|_2$. 
\begin{lem}
    We have $\|\vec{q}_t\|_2 \leq 2(1+\frac{1}{\cheb_t(s(0))})  t$ for all $t \in \N$. In particular, for $t \geq \frac12 \kappa \log(2\kappa^2 / \eps)$ we have $\|\vec q_t\|_2 \leq 2(1+\eps/\kappa^2)t$. 
\end{lem}
\begin{proof}
We start by bounding $\abs{q_t(x)}$ on $[-1, 1]$, and we recall that
\[
    q_t(x) = \frac{1-\cheb_t(s(x))/\cheb_t(s(0))}{x}, \quad\text{where}\quad s(x) = \frac{1+1/\kappa^2 - 2x^2}{1-1/\kappa^2}.
\]
On one hand, when $x \in D_\kappa$ we have $s(x) \in [-1,1]$ and thus $\abs{1-\cheb_t(s(x))/\cheb_t(s(0))} \leq 1 + 1/\cheb_t(s(0))$. On the other hand, when $\abs{x} \leq 1/\kappa$ we have $1 \leq s(x) \leq s(0) = \frac{1+1/\kappa^2}{1-1/\kappa^2}$.
Since $\cheb_t(x)$ is increasing for $x \geq 1$, it follows that $0 \leq 1-\cheb_t(s(x))/\cheb_t(s(0)) \leq 1$ for all $\abs{x} \leq 1/\kappa$. Together this shows that 
\[
\abs{q_t(x)} = \abs{\frac{1-\cheb_t(s(x))/\cheb_t(s(0))}{x}} \leq \frac{1 + 1/\cheb_t(s(0))}{\abs{x}} \qquad \text{for all } x \in [-1,1] \setminus \{0\}.
\]
We now bound the norm of $\vec q_t$. We have
    \begin{align*}
        \norm{\vec{q}_t}^2 &= \sum_{k=1}^{2t} q_t(x_k)^2 \leq \left(1 + \frac{1}{\cheb_t(s(0))}\right)^2 \sum_{k=1}^{2t} \frac{1}{x_k^2} = \left(1 + \frac{1}{\cheb_t(s(0))}\right)^2 \sum_{k=1}^{2t} \frac{1}{\cos^2 \left( \frac{2k-1}{4t}\pi \right)},
    \end{align*}
    where we substituted the exact expression for the Chebyshev nodes $x_k = \cos \left( \frac{2k-1}{4t}\pi \right)$. Moreover, we have
    \[
        \cos^2\left( \frac{2(2t-k+1)-1}{4t}\pi \right) = \cos^2 \left( \frac{2k-1}{4t}\pi \right) = \frac{1-\cos\left(\frac{2k-1}{2t}\pi\right)}{2} \quad\text{for all}\quad 1 \leq k \leq t,
    \]
    where the first equality comes from $x_{2t - k + 1} = -x_k$. Therefore, we have
    \[
        \norm{\vec{q}_t}^2 \leq 4 \left(1 + \frac{1}{\cheb_t(s(0))}\right)^2 \sum_{k=1}^t \frac{1}{1-\cos(\frac{2k-1}{2t}\pi)}.
    \]
    We note that the roots of $\cheb_t(x)$ are exactly $\cos(\frac{2k-1}{2t}\pi)$. For any polynomial $P(x)=C\prod_{k=1}^t (x-r_k)$ the following identity holds for all $x$ for which $P(x) \neq 0$:
    \[
        \sum_{k=1}^t \frac{1}{x-r_k} = \frac{P'(x)}{P(x)}.
    \]
    Applying the above to $P(x)=\cheb_t(x)$ and $x=1$ (which is not a root of $\cheb_t$), we get
    \[
        \sum_{k=1}^t \frac{1}{1-\cos(\frac{2k-1}{2t}\pi)} = \frac{\cheb_t'(1)}{\cheb_t(1)} = \frac{t \cdot \mathcal  U_{t-1}(1)}{1} = t^2.
    \]
    This concludes the main part of the proof: we have shown that $\|\vec q_t\| \leq 2 (1+\frac{1}{\cheb_t(s(0))}) t$. 
    
    Finally, for $t \geq \frac12 \kappa \log(2\kappa^2 / \eps)$, we bound $1/\cheb_t(s(0))$ as in the proof of \cref{thm: cheb approx pos}. Namely, using the same inequalities, we have
    \[
    	\cheb_t(s(0)) \geq \frac12 \left( \frac{\kappa + 1}{\kappa - 1} \right)^t \geq \frac12 \left( 1 + \frac{2}{\kappa - 1} \right)^{\frac12 \kappa \log(2\kappa^2 / \eps)} \geq \frac{\kappa^2}{\eps}.  \qedhere
    \]
\end{proof}
Combining this lemma with \cref{eq:cheb ineq 1-norm}, we derive the same bound for $\norm{\vec{c}_t}_1$:
\begin{cor}\label{cor:bound for ct norm}
	For all $t \in \N$,  $\norm{\vec{c}_t}_1 \leq 2(1+\frac{1}{\cheb_t(s(0))}) t$. In particular, for $t \geq \frac12 \kappa \log(2\kappa^2 / \eps)$ we have $\|\vec c_t\|_1 \leq 2(1+\eps/\kappa^2)t$. 
\end{cor}

\subsection{Efficiently computing the coefficients}
In the case of evaluating $q_t$ via LCU, one question of practical relevance is how to compute the coefficients $\vec{c}_t$. Naively using the recurrence \eqref{eq:qt recurrence} to compute $\vec{c}_t$ gives rise to an algorithm with $O(t^2)$ arithmetic operations with real numbers. Alternatively, one can use FFT-based Chebyshev interpolation algorithms that can compute $\vec{c}_t$ with $O(t \log t)$ operations given the vector $\vec{q}_t$ of the values of $q_t(x)$ at the order-$t$ Chebyshev nodes~\cite{Gentleman1972DCT}. Thus, in order to get an $O(t \log t)$-operation algorithm for computing $\vec{c}_t$, it suffices to show that $q_t(x)$ can be evaluated at a single Chebyshev node $x_k$ with $O(\log t)$-operations. Given the form of $q_t$, this means that we need to compute $\cheb_t(s(x_k))$ with $O(\log t)$ operations. One way to do this is via the degree-halving identities
\[
    \cheb_{2t}(x) = 2\cheb_t(x)^2 - 1 \quad\text{and}\quad \cheb_{2t+1}(x) = 2\cheb_{t+1}(x)\cheb_t(x) - x.
\]

\subsection{A more natural quantum algorithm?}
Given the reduction of the general linear system problem to the PD case (\cref{cor:reducing general linear systems to psd}), one might be tempted to mirror this reduction when designing a quantum algorithm, with the goal of achieving $\tilde O(\sqrt{\kappa})$ complexity for solving PD systems. The input of such an algorithm would be a (block-encoding of a) Hermitian matrix $\mat A $ with eigenvalues in $[1/\kappa, 1]$, and the output would be a block-encoding of $q_t^+(\mat A)$. To evaluate this polynomial using QSVT, we first need to normalize it by dividing it by $\max_{x\in [-1, 1]} |q_t^+(x)|$. It turns out that this maximum grows exponentially with $t$: one can lower bound it by $|q_t^+(-1)|$ and we have
\begin{align*}
	|q_t^+(-1)| &\geq \frac{\cheb_t\left(\frac{1+1/\kappa+2}{1-1/\kappa}\right)}{\cheb_t\left(\frac{1+1/\kappa}{1-1/\kappa}\right)} - 1 \geq \frac{\cheb_t( 3+4/(\kappa - 1) )}{\cheb_t(1 + 2/(\kappa - 1))} \\
	&\geq \frac{\cheb_t(3)}{\cheb_t(2)} \geq \frac12 \left(\frac{3 + 2\sqrt{2}}{2+\sqrt{3}}\right)^t \geq \frac12 \left(\frac32\right)^t.
\end{align*}
Therefore, amplifying the output of QSVT would take exponential time. In the case of LCU, the coefficient 1-norm is lower bounded by $|q_t^+(-1)|$ (by \cref{eq:coeff norm is an upper bound for p on the interval}), so the output of a LCU-based algorithm would also need to be amplified exponentially. Alternative approaches of multiplying $q_t^+(x)$ by a rectangle function that is close to $1$ on $[1/\kappa, 1]$ and close to $0$ elsewhere are similarly fruitless as the degree of the resulting approximation polynomial would become linear in $\kappa$. It should be noted, however, that these issues can be avoided if we assume that the mapping $x \mapsto \frac{1+1/\kappa - 2x}{1-1/\kappa}$ has already been performed ``ahead of time'': in \cite{orsucci2021solving}, Orsucci and Dunjko have shown that PD matrices can indeed be inverted in $\softO(\sqrt{\kappa})$, provided that a block-encoding of $I - \alpha \mat A$ is given as input (for suitable $\alpha$).

Another natural alternative approach would be to quantize a method such as momentum gradient descent, which also converges in $\softO(\sqrt{\kappa})$ for PD matrices \cite{Polyak1987Optimization}. One way to achieve this would be using the approach of Kerenidis and Prakash \cite{KP20Gradient}, who quantized the basic gradient descent algorithm by implementing the recurrence $\vec{r}_{t+1} = (I - \eta \mat A)\vec{r}_t$ satisfied by the differences $\vec{r}_t := \vec{x}_t - \vec{x}_{t-1}$ of successive iterates. Applying this idea to momentum gradient descent, one gets a recurrence involving two successive differences:
\[
    \begin{bmatrix}
            \vec{r}_{t+1} \\
            \vec{r}_t
    \end{bmatrix} = \underbrace{\begin{bmatrix}
            (1 + \beta) I - \eta \mat A & -\beta I \\
            I & 0
    \end{bmatrix}}_{\mat M} \begin{bmatrix}
            \vec{r}_t \\
            \vec{r}_{t-1}
    \end{bmatrix}, 
\]
for suitable choices of $\eta$ and $\beta$. For example, following \cite[Chapter 3]{Polyak1987Optimization}, one can set $\eta = 4/(1+\sqrt{1/\kappa})^2$ and $\beta = \left( 1-2/(1+\sqrt{\kappa}) \right)^2$. 
Implementing a similar approach as in \cite{KP20Gradient} would require the construction of $O(1)$-block-encodings of powers of $\mat M$. In particular, this would require $\mat M$ to have a small norm. Unfortunately, for large enough $\kappa \geq 9$ and the above choice of $\eta,\beta$, one has $\|\mat M\| \geq \sqrt{2}$ which means that a block-encoding of $\mat M^t$ needs to have sub-normalization at least $2^{t/2}$. 

\section{Query lower bounds}\label{sec:Query lower bound}
So far, we have been considering algorithms (i.e. upper bounds) for the QLS problem. The complexity of the best algorithm for the QLS problem
depends linearly on $\kappa$ (we ignore the polylogarithmic factors in this section), so a natural question is whether this dependence is optimal. In \cite{HHL09} it has been shown that this is indeed the case: in the sparse access input model (the setting in which such lower bounds are usually proven), the complexity of QLS for general systems is $\Omega(\min(\kappa,n))$. Recently, it has been shown \cite{orsucci2021solving} that the same $\Omega(\min(\kappa,n))$ lower bound even holds for the restriction of QLS to PD matrices -- this is surprising since in the classical setting a $\sqrt{\kappa}$-separation exists between the general and the PD case. We note that both of these lower bounds apply when the output of the QLS solver is the quantum state $\ket{\mat A^{-1} \vec{b}}$. As a consequence, one can show that computing a classical description of $\mat A ^{-1}\vec{b}$ is just as hard.

Both of the above results apply to the small-$\kappa$ regime. In particular, they leave open the possibility of a $o(n^\omega)$-time quantum algorithm for solving linear systems (with classical output). The existence of such an algorithm would speed up many classical optimization algorithms (e.g., interior point methods) in a black-box way. 
In \cite{Dorn2009QueryComplexity} it was shown that one cannot obtain a large quantum speedup when the output is required to be classical: $\Omega(n^2)$ quantum queries to the entries of $\mat A $ are needed to obtain a classical description of a single coordinate of $\mat A ^{-1} e_n$, where $e_n$ is the $n$-th standard basis vector in $\R^n$. The statement is robust in the following sense: after normalizing $\mat A ^{-1}e_n$, it suffices to obtain a $\delta$-additive approximation of the first coordinate for some $\delta = O(1/n^2)$. We present a simplified proof of this result of \cite{Dorn2009QueryComplexity} at the end of this section.
Note that this high precision prevents one from lifting the bound to the quantum-output setting: to obtain a $\delta$-additive approximation of a single coordinate of $\ket{\mat A^{-1}b}$ one can use roughly $1/\delta$ rounds of amplitude estimation on a QLS-solver~$\mathcal A$. With $\delta=O(1/n^2)$ this only implies that $n^2 \cdot \text{cost}(\mathcal A) = \Omega(n^2)$. A second type of quantum lower bound is described in \cite{Gilyen2019SVT}: roughly speaking, if a (smooth) function $f:I \to [-1,1]$ has a derivative whose absolute value is $d$, then $\Omega(d)$ uses of a $1$-block-encoding $U_{\mat A}$ of $\mat A$ are needed to create a block-encoding of $f(\mat A)$. Here $I$ is a subset of $[-1,1]$ that contains the eigenvalues of the Hermitian matrix $\mat A$. Applied to $f(x) = 1/(\kappa x)$, this shows that indeed $\Omega(\kappa)$ applications of $U_{\mat A}$ are needed to create a block-encoding of $\mat A^{-1}$. As mentioned before, a block-encoding of $\mat A^{-1}$ can be combined with a state preparation oracle for $\vec b$ to solve the QLS problem. Such a strategy however naturally incurs a $\kappa$-dependence in the runtime, and it remains an interesting open question whether one could solve the QLS problem (with quantum output!) without such a dependence in $\kappa$ and in time $o(n^\omega)$. 

\subsection{Lower bound for matrix inversion with classical output} \label{app: Lower bound}
We present a simplified proof of a matrix-inversion lower bound result of~\cite{Dorn2009QueryComplexity}. It is based on the quantum query complexity of the majority function $\mathrm{MAJ}_n:\{0,1\}^n \to \{0,1\}$ which is takes value $1$ on input $\vec x$ if and only if $\sum_{i \in [n]} x_i > n/2$. It is well known that the quantum query complexity of $\mathrm{MAJ}_n$ is $\Theta(n)$~\cite{Beals01PolynomialMethod}. 
\begin{lem}\label{lem:Matrix power lower bound}
    Let $\mat X \in \{0, 1\}^{n\times n}$. Then, the matrix $\mat A  \in \{0, 1\}^{(2n+2) \times (2n+2)}$ defined as
    \[
        \mat A = \left\lbrack\begin{array}{@{}c|c@{}}
  \begin{matrix}
 0  & 1_n^* \\
 1_n & 0 
  \end{matrix}
  &   \begin{matrix}
 0  & 0 \\
 \mat X & 0 
  \end{matrix} \\ \hline 
    \begin{matrix}
 0  & \mat X^* \\
 0 & 0 
  \end{matrix} &
    \begin{matrix}
 0  & 1_n \\
 1_n^* & 0 
  \end{matrix}
\end{array}\right\rbrack
    \]
    satisfies $(\mat A^3)_{1, 2n+2} = \sum_{i=1}^n \sum_{j=1}^n X_{i,j}$.
\end{lem}
\begin{proof}
    $\mat A $ is the adjacency matrix of an undirected graph that can be described as follows. We start with a bipartite graph between two sets of $n$ vertices whose edge set is described by $X$, then we add two vertices labeled $1$ and $2n+2$ that we connect respectively to the first set of vertices and the second set of vertices. 
    The entry $(1, 2n+2)$ of $\mat A ^3$ counts the number of paths of length $3$ from $1$ to $2n+2$ in this graph. This equals the number of edges between the sets $\{2, \dots, n+1\}$ and $\{n+2, \dots, 2n+1\}$, that is, $(\mat A^3)_{1,2n+2} = \sum_{i=1}^n \sum_{j=1}^n X_{i,j}$.
\end{proof}
\begin{cor}
Let $\mat A  \in \{0,1\}^{n \times n}$. Determining a single off-diagonal entry of $\mat A ^{3}$, with success probability $\geq 2/3$, takes $\Theta(n^2)$ quantum queries to $\mat A $.
\end{cor}
\begin{lem}
    Let $\mat A  \in \{0, 1\}^{n \times n}$. Then, for $N = 4n$, the matrix $\mat B \in \{ 0,1 \}^{N \times N}$ defined by 
    \[
        \mat B = \begin{bmatrix}
            I & \mat A && \\
            & I & \mat A & \\
            && I & \mat A \\
            &&& I
        \end{bmatrix},
    \]
    satisfies $(\mat B^{-1})_{1,N} = -(\mat A^3)_{1,n}$.
\end{lem}
\begin{proof}
    It is straightforward to verify that the inverse of $\mat B$ is 
    \[
        \mat B^{-1} = \begin{bmatrix}
            I & -\mat A & \mat A^2 & -\mat A^3 \\
              &  I & -\mat A  &  \mat A^2 \\
              &    &  I  & -\mat A   \\
              &    &     &  I   \\
        \end{bmatrix}. \qedhere
    \]
\end{proof}
If $\mat A $ is the adjacency matrix of the directed version of the graph described in Lemma~\ref{lem:Matrix power lower bound}, we can also compute the norm of the last column as follows:
\[
    \norm{B^{-1} \vec{e}_{4n}}^2 = \left( \sum_{i,j} X_{i,j} \right)^2 + \sum_i \left(\sum_j X_{i,j} \right)^2 + n+1.
\]
In particular, for the hard instances (where $|n/2 - \sum_{i,j} X_{i,j}| \leq 1$), we have that $\norm{\mat B^{-1} \vec{e}_{4n}} = \Theta(n^2)$.
\begin{cor}
Let $\mat A  \in \{0,1\}^{n \times n}$. Determining a single off-diagonal entry of $\mat A ^{-1}$ up to precision $<1/2$, with success probability $\geq 2/3$, takes $\Theta(n^2)$ quantum queries to $\mat A $.
\end{cor}

\bibliographystyle{plain}
\bibliography{bibliography}

\appendix 

\section{Examples of functions with bounded Chebyshev coefficient norms}\label{app:Good example functions for LCU}
The inverse function is not the only function that can be efficiently evaluated using LCU of Chebyshev polynomials. Here we discuss several families of functions for which the 1-norm of the Chebyshev coefficients is of the order $\log(\text{degree})$.

\subsection{Simple examples}\label{app:Simple functions with bounded coeffs}
We first observe that the monomial $x^n$ has the following Chebyshev expansion:
\[
    x^{n}=2^{1-n}\mathop {{\sum }'} _{j=0,\,n-j\,\mathrm {even} }^{n}{\binom {n}{\tfrac {n-j}{2}}}\cheb_{j}(x),
\]
where the prime at the sum symbol indicates that the contribution of $j = 0$ needs to be halved (if it appears). The sum of these coefficients is bounded by 1. This implies that for any polynomial the $1$-norm of the coefficients in the Chebyshev basis is at most the $1$-norm of the coefficients in the monomial basis. 
This means, for example, that the Chebyshev coefficient 1-norm of the scaled exponential is at most $1$. Similarly, for a degree $n$ approximation of the (scaled) logarithm the $1$-norm grows as $O(\log n)$. In particular, they have the following Taylor expansions for $\kappa \geq 1$
\begin{align*}
    e^{\kappa (x-1)} = e^{-\kappa} \sum_{n=0}^\infty \frac{(\kappa x)^n}{n!},
\end{align*}\begin{align*}
    \operatorname{slog}_\kappa(x) &:= \log( 1/\kappa + ((x+1)/2)(1-1/\kappa) ) = \log\left( \frac{\kappa+1}{2\kappa} \left( 1+\frac{\kappa - 1}{\kappa + 1}x \right) \right) \\
    &= \log\left( \frac{\kappa+1}{2\kappa} \right) + \sum_{n=1}^\infty \frac{(-1)^{n+1}}{n} \left( \frac{\kappa - 1}{\kappa + 1} \right)^n x^n.
\end{align*}

\subsection{Approximating discontinuities -- the error function}\label{app:The error function}
Some more interesting examples are the sign and the rectangle functions, defined as 
\[
    \sign(x) := \begin{cases}
        1 & \text{if } x>0,\\
        0 & \text{if } x=0,\\
        -1 & \text{if } x<0,
    \end{cases} \quad\text{and}\quad
    \Pi(x) := \begin{cases}
        1 & \text{if } |x| \leq 1/2,\\
        0 & \text{else}.
    \end{cases}
\]
It is well-known \cite{Low2017UniformAmp, Gilyen2019SVT} that the error-function $\erf(x) = \frac{2}{\sqrt{\pi}}\int_{0}^x e^{-z^2} \dif z$ is a fundamental building block for approximating discontinuous functions. For example, given $\epsilon, \delta > 0$, there exists a choice of $\kappa = O(\polylog(1/\epsilon)/\delta)$ such that $\erf(\kappa x)$ is $\epsilon$-close to $\sign(x)$ on $[-1,1] \setminus [-\delta, \delta]$. We show below that the $1$-norm of the coefficients of the Chebyshev series of $\erf(\kappa x)$ is $O(\log \kappa)$. We start with the following expansion from \cite{Low2017UniformAmp}:
\[
    \erf(\kappa x) = \frac{2\kappa e^{-\kappa^2/2}}{\sqrt{\pi}} \left( I_0(\kappa^2/2)x + \sum_{j=1}^{\infty} I_j(\kappa^2/2)(-1)^j \left( \frac{\cheb_{2j+1}(x)}{2j+1} - \frac{\cheb_{2j-1}(x)}{2j-1} \right)\right).
\]
By regrouping the terms, we get the following explicit form of the Chebyshev series of $\erf(\kappa x)$:
\begin{equation}\label{eq:erf coefficients}
    \erf(\kappa x) = \frac{2\kappa e^{-\kappa ^2/2}}{\sqrt{\pi}} \sum_{n=0}^\infty (-1)^n \frac{I_n(\kappa^2/2) + I_{n+1}(\kappa^2/2)}{2n+1} \cheb_{2n+1}(x).
\end{equation}
Now, in order to bound the coefficient norm, we use the following inequality from \cite{Baricz2014Bessel}:
\[
    e^{-x} x^{-n} (I_n(x) + I_{n+1}(x)) \leq \sqrt{\frac{2}{\pi}} \left( x+ \frac{n}{2} + \frac14 \right)^{-n-\frac12}.
\]
Note that $I_n(x) \geq 0$ for $x \geq 0$ and all $n \in \N$. So, the above in fact bounds the absolute value of the left hand side. We use this inequality to bound the (absolute value of the) coefficient of $\cheb_{2n+1}(x)$ in \eqref{eq:erf coefficients} as follows:
\begin{equation}\label{eq:erf coefficient bound}
    \frac{2\kappa e^{-\kappa^2/2}}{\sqrt{\pi}} \frac{I_n(\kappa^2/2) + I_{n+1}\kappa^2/2)}{2n+1} \leq \frac{4}{\pi} \frac{1}{2n+1} \left( \frac{\kappa^2}{\kappa^2 + n+ 1/2} \right)^{n+1/2}.
\end{equation}
Using this inequality, we can bound the coefficient norm of the truncated Chebyshev series: 
\begin{lem}
    Let $N > 0$ be an integer. Then, 
    \[
        \frac{2\kappa e^{-\kappa^2/2}}{\sqrt{\pi}} \sum_{n=0}^N \frac{I_n(\kappa^2/2) + I_{n+1}(\kappa^2/2)}{2n+1} \leq \frac{6 + 2\log N}{\pi}.
    \]
\end{lem}
\begin{proof}
    Using \eqref{eq:erf coefficient bound}, and the fact that $0 \leq \frac{\kappa^2}{\kappa^2 + n+ 1/2} \leq 1$, we get 
    \begin{align*}
        \frac{2\kappa e^{-\kappa^2/2}}{\sqrt{\pi}} \sum_{n=0}^N \frac{I_n(\kappa^2/2) + I_{n+1}(\kappa^2/2)}{2n+1} &\leq \frac{4}{\pi}\sum_{n=0}^N  \frac{1}{2n+1} \left( \frac{\kappa^2}{\kappa^2 + n+ 1/2} \right)^{n+1/2} \\
        &\leq \frac{4}{\pi}\sum_{n=0}^N  \frac{1}{2n+1}.
    \end{align*}
    It is well-known that the last sum is $O(\log N)$. To be more precise,
    \begin{align*}
        \frac{4}{\pi}\sum_{n=0}^N  \frac{1}{2n+1} \leq \frac{4}{\pi} \left(1 + \sum_{n=1}^N \frac{1}{2n}\right) = \frac{4}{\pi} + \frac{2}{\pi} \sum_{n=1}^N \frac{1}{n} \leq \frac{4}{\pi} + \frac{2}{\pi} (1+\log N) \leq \frac{6 + 2\log N}{\pi}.
    \end{align*}
\end{proof}
Now, if we just want to bound the coefficients' 1-norm, it suffices to take $N = \lceil \kappa^2 \rceil$, and bound the rest of the coefficients using the following simple tail bound:
\begin{lem}
    Let $N \geq \kappa^2$ be an integer. Then,
    \[
        \frac{2\kappa e^{-\kappa^2/2}}{\sqrt{\pi}} \sum_{n=N}^\infty \frac{I_n(\kappa^2/2) + I_{n+1}(\kappa^2/2)}{2n+1} \leq 2^{2-N}.
    \]
\end{lem}
\begin{proof}
    Again, we start by using \eqref{eq:erf coefficient bound}, but now we note that for $n \geq \kappa^2$, $0 \leq \frac{\kappa^2}{\kappa^2 + n+ 1/2} \leq \frac12$, so
    \begin{align*}
        \frac{2\kappa e^{-\kappa^2/2}}{\sqrt{\pi}} \sum_{n=N}^\infty \frac{I_n(\kappa^2/2) + I_{n+1}(\kappa^2/2)}{2n+1} &\leq \frac{4}{\pi}\sum_{n=N}^\infty \frac{2^{-n}}{2n+1} \leq \frac{2^{3-N}}{\pi} \leq 2^{2-N}. \qedhere
    \end{align*}
\end{proof}
Therefore, the coefficient norm of the entire series is bounded by $\frac{6 + 2\log \kappa^2}{\pi} + 2^{2-\kappa^2} \leq 4 + 2\log \kappa$. An easy consequence of this bound is that we can approximate $\erf(\kappa x)$ up to error $0 \leq \epsilon \leq 2^{2-\kappa^2}$ with a polynomial of degree $\log_2(4/\epsilon)$.

If the desired error $\epsilon$ is larger than $2^{2-\kappa^2}$, a more careful analysis of the tail bound for $\kappa \leq N \leq \kappa^2$ yields an $\epsilon$-approximation polynomial of degree $O(k\sqrt{\log(\kappa/\epsilon)})$.
\begin{lem}
    Let $1 \leq \alpha \leq \kappa$ be an integer. Then, 
    \[
        \frac{4}{\pi} \sum_{n=\alpha \kappa}^{(\alpha+1)\kappa - 1} \frac{1}{2n+1} \left( \frac{\kappa^2}{\kappa^2 + n+ 1/2} \right)^{n+1/2} \leq \frac{4}{\pi} e^{\alpha^2 / 2}.
    \]
\end{lem}
\begin{proof}
    First, we note that $\left(\frac{\kappa^2}{\kappa^2+n+1/2}\right)^{n+1/2} \leq \left(\frac{\kappa^2}{\kappa^2+n}\right)^{n}$, so we get
    \begin{align*}
        \frac{4}{\pi} &\sum_{n=\alpha \kappa}^{(\alpha+1)\kappa - 1} \frac{1}{2n+1} \left( \frac{\kappa^2}{\kappa^2 + n+ 1/2} \right)^{n+1/2} \leq \frac{4}{\pi} \sum_{n=\alpha \kappa}^{(\alpha+1)\kappa -1}  \frac{1}{2n+1} \left(\frac{\kappa^2}{\kappa^2+n}\right)^{n} \\
        &\leq\frac{4}{\pi}  \sum_{n=\alpha \kappa}^{(\alpha+1)\kappa -1}  \frac{1}{2(\alpha \kappa)+1} \left(\frac{\kappa^2}{\kappa^2+\alpha \kappa}\right)^{\alpha \kappa} = \frac{4}{\pi}  \frac{\kappa}{2(\alpha \kappa)+1} \left(\frac{\kappa}{\kappa+\alpha}\right)^{\alpha \kappa} \\
        &= \frac{4}{\pi}  \frac{\kappa}{2(\alpha \kappa)+1} \left(\frac{1}{1+\alpha/\kappa}\right)^{\alpha \kappa} \leq \frac{4}{\pi}  \frac{\kappa}{2(\alpha \kappa)+1} \left(e^{-\alpha/(2\kappa)}\right)^{\alpha \kappa} \\
        &= \frac{4}{\pi}  \frac{\kappa}{2(\alpha \kappa)+1} e^{-\alpha^2/2} \leq \frac{4}{\pi}  e^{-\alpha^2/2}.
    \end{align*}
    The second to last inequality requires $\alpha \leq \kappa$.
\end{proof}
So, to get an $\epsilon$-approximation polynomial, we just need to an integer $1 \leq \alpha_0 \leq \kappa$ such that \[
    2^{2-\kappa^2} + \frac{4}{\pi} \sum_{\alpha = \alpha_0}^\kappa e^{-\alpha^2/2} \leq \epsilon.
\]
Indeed, if we let $\epsilon' = \epsilon - 2^{2-\kappa^2}$, it suffices to choose $\alpha_0 = \left\lceil \sqrt{2\log(\frac{4 \kappa}{\pi\epsilon'})}\right\rceil$, so we get
\[
    2^{2-\kappa^2} + \frac{4}{\pi} \sum_{\alpha = \alpha_0}^\kappa e^{-\alpha^2/2} \leq 2^{2-\kappa^2}+ \frac{4\kappa}{\pi} e^{-\alpha_0^2 / 2} \leq 2^{2-\kappa^2} + \epsilon' = \epsilon.
\]
Thus, the degree of the $\epsilon$-approximating polynomial is $\alpha_0 \kappa -1 = O(\kappa \sqrt{\log(\kappa/\epsilon)})$.

\end{document}